\documentclass{article} % For LaTeX2e
\usepackage{nips15submit_e,times}

\usepackage{hyperref}
\usepackage{url}
\usepackage[utf8]{inputenc}
\usepackage{amsmath,amssymb}
\newcommand{\vertiii}[1]{{\left\vert\kern-0.25ex\left\vert\kern-0.25ex\left\vert #1 
    \right\vert\kern-0.25ex\right\vert\kern-0.25ex\right\vert}}
\usepackage{bbm}
\usepackage{graphicx}
\usepackage{subfig}
\usepackage{caption}
\usepackage{algorithm}
\usepackage{algorithmic}
\usepackage{epsfig}
\usepackage{amsthm}
\usepackage{breqn} %break equations
\usepackage{changebar}
\usepackage{epstopdf}
\usepackage{tabularx}
\usepackage{wrapfig}
\usepackage{enumitem}
\usepackage{mathrsfs}

\newtheorem{theorem}{Theorem}
\newtheorem{remark}{Remark}

\newtheorem{example}{Example}

\newtheorem{definition}{Definition}

\usepackage{amssymb,amsmath}

\input xy
\xyoption{all}

\begin{document}

\title{Identifying Nonlinear 1-Step Causal Influences in Presence of Latent Variables}

\maketitle
\begin{abstract}
We propose an approach for learning the causal structure in stochastic dynamical systems with a $1$-step functional dependency in the presence of latent variables.
We propose an information-theoretic approach that allows us to recover the causal relations among the observed variables as long as the latent variables evolve without exogenous noise. We further propose an efficient learning method based on linear regression for the special sub-case when the dynamics are restricted to be linear. We validate the performance of our approach via numerical simulations. 
\end{abstract}

\section{Introduction}
% no \IEEEPARstart
Identifying causal influences in a network of time series is one of fundamental problems in many different fields, including social sciences, economics, computer science, and biology. In macroeconomics, for instance, researchers seek to understand what are the factors contributing to economic fluctuations and how these factors interact with each other \cite{Lutkepohl04}. In neuroscience, extensive body of research focuses on learning the interactions between different regions of brain by analyzing neural spike trains \cite{Quinn}.% In all these applications, the natural question is whether the causal relations can be extracted just by observing a set of time-series. But any answer to this question, first needs a rigorous definition of causality. 

In 1960's, Granger proposed a definition of causality between random processes \cite{Granger}. The key idea of his definition is that if a process $X_2(t)$ causes another process $X_1(t)$, then knowing the past of $X_2(t)$ up to time $t$ must aid in predicting $X_1(t)$. In particular, let $\Sigma_{X_1}(h|\Omega_t)$ be the mean square error (MSE) of the optimal $h$-step predictor of a random process $X_1(t)$ at time $t$ given information $\Omega_t$. Process $X_2(t)$ is said to Granger cause process $X_1(t)$ if:
\begin{equation}
\exists h>0, \mbox{ $s.t.$ } \Sigma_{X_1}(h|\Omega_t)<\Sigma_{X_1}(h|\Omega_t\backslash \{X_2(s)\}_{s=0}^t),
\end{equation}
where the set $\Omega_t$ contains all information in the universe related to the past and the present of $X_1(t)$. We also say that the process $X_2(t)$ has a 1-step cause on $X_1(t)$ if the above inequality holds for $h=1$. In other words, considering $X_2(t)$ in the set $\Omega_t$ improves prediction of $X_1(t+1)$. 

Granger's definition of causality is consistent with the belief that a cause cannot come after the effect, but it is not practical in some settings because it requires knowledge of the entire set $\Omega_t$. To put it differently, it is hard to identify and account for all parts of universe that are related to a specific process $X_1(t)$. Hence, only the available information related to $X_1(t)$ is considered in practice \cite{Lutkepohl}. To see what may go wrong in such a situation, consider the following linear model with three state variables:
\begin{equation}
\begin{bmatrix}
X_1(t)\\X_2(t)\\Z(t)
\end{bmatrix}
=\begin{bmatrix}
0& 0&0.5\\
0.5&0.1&0.9\\
0.9&0&0.5
\end{bmatrix}
\begin{bmatrix}
X_1(t-1)\\X_2(t-1)\\Z(t-1)
\end{bmatrix}
+
\begin{bmatrix}
\omega_{1}(t)\\\omega_{2}(t)\\0
\end{bmatrix},
\end{equation}
where $X_1(t)$ and $X_2(t)$ are observable but $Z(t)$ is latent. Let $\omega_{1}(t)$ and $\omega_{2}(t)$ be i.i.d random variables with the same variance. If we fit a linear model only on $X_1(t)$ and $X_2(t)$ without considering $Z(t)$, our estimation of upper left $2\times 2$ submatrix would be: $[0.06,0.32;0.61,0.69]$. This result implies that $X_2(t)$ is a 1-step cause of $X_1(t)$ with the strength $0.32$ which is wrong. The concept of Granger causality can be generalized to nonlinear setting using an information-theoretic quantity ``directed information" \cite{Marko}. Still the problem caused by latent processes persists in that setting as well.

% Again, this quantity may be misleading if latent processes are not taken into account in computing directed information.

Identifying causal relations between random variables has been studied in the presence of latent variables to some extent. For instance, Elidan et al. proposed an algorithm based on expectation maximization (EM) to estimate the parameters of their model by fixing the number of latent variables and also the structural relationships between latent and observed variables \cite{elidan2007ideal}. 
 Chandrasekaran et al. \cite{chandrasekaran2010latent} presented a tractable convex program based on regularized maximum likelihood for recovering causal relations for a model where the latent and observed variables are jointly Gaussian, and the conditional statistics of the observed variables given the latent variables is a sparse graph. A well-known approach for learning latent Markov models uses quartet-based distances to discover the structure \cite{jiang,erdos}. In most of quartet-based solutions, a set of quartets is constructed for all subsets of four observable variables and then quartets are merged to form a tree structure.

In recent years, there has been an increasing interest in inferring causal relations in random processes. Jalali and Sanghavi showed that 1-step causal relations between observed variables can be identified in a Vector Auto-Regressive (VAR) model assuming that connections between observed variables are sparse and each latent variable interacts with many observed variables \cite{Jalali}. In \cite{Geiger}, Geiger et al. showed that identifying 1-step causes between observed variables is possible under some algebraic conditions on the transition matrix of VAR model. Recently, Etesami et al. studied a network of processes with polytree structure and introduced an algorithm that can learn latent polytrees using a discrepancy measure \cite{Etesami}.

In this paper, we propose an information-theoretic criteria for identifying the causal relations in a general model of stochastic dynamical systems without restricting the mapping functions (say to linear mappings) or the underlying structure (e.g., being a tree) among the observed nodes also when there is no exogenous noise in the latent part.
%In this work, we introduce $1$-step functional dependency as a notion of causality and show that all such functional dependencies among the observed processes can be detected when there is no exogenous noise in the latent part.
We propose an efficient method to identify functional dependencies for the special case of linear mappings. We further demonstrate the applicability of this criteria though simulation results for both linear and nonlinear cases.

The paper is organized as follows. In Section \ref{sec:def}, we provide the preliminary definitions and describe the system model. In Section III, we present the main result and study the special restriction of it to linear models. We provide our simulation results in Section IV. Finally, we conclude in Section V.

\vspace{.3cm}

\section{Problem Definition}\label{sec:def}
In this section, after some notational conventions, the model of stochastic dynamical system is presented. Afterwards, we present our definition of 1-step functional dependency between the processes for this model.

\subsection{Notations}
Any $n\times 1$ vector with with entries $[V_1(t);\cdots;V_n(t)]$ is denoted by $\vec{V}(t)$. 
We denote the $t$-th random variable in the $i$-th process by $V_{i}(t)$. We use underlined characters to represent a collection of processes, for example $\underline{V}_{\mathcal{K},0}^t$ is used to denote a set of random processes with index set $\mathcal{K}$ from time $0$ up to time $t$. For $\mathcal{K}=\{1,...,n\}$, we denote $\underline{V}_{\mathcal{K},0}^t$ by $\underline{V}_{0}^t$. We also define: $-j:=\{1,\cdots,n\}\backslash \{j\}$. The identity matrix of size $n$ is shown by $I_{n\times n}$. We denote $(i,j)$ entry of a matrix $A$ by $A(i,j)$. 

In a directed graph $\overrightarrow{G}=(V,\overrightarrow{E})$ that is characterized by a set $V$ of vertices (or nodes) and a set of ordered pairs of vertices, called arrows (or edges) $\overrightarrow{E}\subset V\times V$, we denote the set of \textit{parents} of a node $v$ by $\mathcal{PA}(v)$ and define it as 
$
\mathcal{PA}(v):=\{u\in V: (u,v)\in\overrightarrow{E}\}
$.

%The Frobenius norm of matrix $A_{n\times m}$  is given by 
%\begin{small}$||A||^2_F:=\displaystyle\sum_{i=1}^{n}\displaystyle%\sum_{j=1}^{m}|A(i,j)|^2$\end{small}. 

\subsection{System model}
Consider a dynamical system described by $n+m$ states
$\vec{V}(t)=[V_1(t),\cdots,V_{n+m}(t)]$ in which the first $n$ processes, denoted by $\vec{X}=[X_1,...,X_n]$ are observable states and the rest which denoted by $\vec{Z}=[Z_1,...,Z_m]$ are latent. More precisely, the joint dynamic of the system is given by:
\begin{eqnarray}\label{Maineq}
\left\{
\begin{array}{ll}
X_i(t)= F_i(\vec{X}(t-1),\vec{Z}(t-1))+\omega_i(t),& 1\leq i\leq n,\\
\vec{Z}(t)=G(\vec{X}(t-1),\vec{Z}(t-1)),&
   \end{array} 
\right. 
\end{eqnarray}
where the exogenous noises $\{\omega_i(t)\}$ are i.i.d. with mean zero. $F_i:\mathbb{R}^{n+m}\rightarrow \mathbb{R}$, $G:\mathbb{R}^{n+m}\rightarrow \mathbb{R}^m$ are mapping functions that belong to appropriately constrained class of functions. Furthermore, we assume that $\vec{Z}(0)$ is a vector of unknown but fixed values. 
The goal of this work is to identify the causal structure among the observed processes $\vec{X}$ given their realizations. Next, we formally introduce what we mean by a causal structure of a dynamical system.

\subsection{Causal Structural Graph}

In dynamical systems with functional dependencies, there is a natural notion of influence among the processes, in the sense that process $V_j$ causes process $V_i$, if $V_i$ is a function of $V_j$. Such dependencies has been studied in the literature \cite{etesami2016measuring}.  
Adopting the definition of functional dependency in \cite{etesami2016measuring}, we define the causal structure of the system in (\ref{Maineq}) as follows.

Random process $V_i$ 1-step functionally depends on process $V_j$ over the time horizon $[0,T]$, if changing the value of $V_j(t-1)$ while keeping all the other variables fixed results in a change in $V_i(t)$ for some time $0<t\leq T$. Next, we present our formal definition of functional dependencies in systems whose joint dynamics is described by (\ref{Maineq}).

\begin{definition}
We say $V_j$ 1-step functionally influences $V_i$ if and only if $\alpha_{i,j}:=1/T\sum_{t=0}^{T}\alpha_{i,j}(t)>0$, where 
\begin{equation}\label{eq:alpha}
\alpha_{i,j}(t):=\!\!\!\sup_{\substack{v_j, v'_j, \vec{v}_{-j}}}|F_i(v_j,\vec{v}_{-j})-F_i(v^{\prime}_j,\vec{v}_{-j})|,
\end{equation} 
$(v_j,\vec{v}_{-j})$ and $(v'_j,\vec{v}_{-j})$ are two realizations of $(V_j(t-1),\vec{V}_{-j}(t-1))$.
\end{definition}
In order to visualize the causal structure in (\ref{Maineq}), we introduce a directed graph whose nodes represent random processes and there is an arrow from node $j$ to nodes $i$, if $V_j$ 1-step functionally influences $V_i$. 
 \begin{example}\label{ex:1}
Consider a causal system with 3 processes such that their joint dynamic is given by:
\begin{align}\notag
&X_1(t)= e^{-|X_1(t-1)+X_2(t-1)|}-e^{-|Z(t-1)|}/5+\omega_1(t),\\ \notag
&X_2(t)= \sqrt{|X_2(t-1)|/2}+\omega_2(t),\\ \notag
&Z(t)= X_2(t-1)Z(t-1),
\end{align}
where $\omega_i$s are independent exogenous noises. Figure \ref{fig:3345} depicts the functional dependency graph of this system.
\begin{figure}
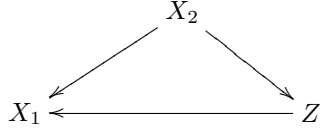

\hspace{5cm}
\xygraph{ !{<0cm,0cm>;<1.9cm,0cm>:<0cm,1.5cm>::} 
!{(-.7,1) }*+{X_{1}}="x1"
!{(.4,1.9) }*+{X_{2}}="x2"
 !{(1.3,1)}*+{Z}="z"
"x2":"x1" "x2":"z"  "z":"x1" }
   \caption{Functional dependency graph of Example \ref{ex:1}.}\label{fig:3345}
\end{figure}
\end{example}

Directed Information Graphs (DIGs) are another type of graphical models that encode statistical dependencies in dynamical systems \cite{dimitrov2011information}. These graphs are defined using an information-theoretic measure, the ``conditional directed information" \cite{massey1990causality,quinn2013efficient}. 
The relationship between the functional dependencies in a stochastic dynamical system and their corresponding DIG has been studied in \cite{etesami2016measuring}.  

For the sake of completeness, we present the definition of DIG. Consider two random processes $V_{i}$ and $V_{j}$ and a set of indices $\mathcal{K}$ such that $\mathcal{K}\subseteq\{1,...,n\}\setminus\{i,j\}$, then the conditional directed information from $V_{j}$ to $V_{i}$, given $\underline{V}_{\mathcal{K}}$ is defined as:
%\begin{small}
\begin{align}\label{cdif}
I(V_{j}\rightarrow V_{i}||\small{\underline{V}}_{\mathcal{K}}):=\mathbb{E}_{P_{\small{\underline{V}}_{\mathcal{K}\cup\{i,j\}}}}\left[\log \dfrac{dP_{V_{i}||V_{j},\small{\underline{V}}_{\mathcal{K}}}}{dP_{V_{i}||\small{\underline{V}}_{\mathcal{K}}}}\right],
\end{align}
%\end{small}
where $\frac{dP_{\small{{\underline{V}}}}}{dQ_{\small{{\underline{V}}}}}$ is the Radon-Nikodym derivative \cite{royden1988real} and $P_{V_i||\underline{V}_{\mathcal{K}}}$ denotes the causal conditioning defined as
%\begin{small}
$
P_{V_{i}||\underline{V}_{\mathcal{K}}}:=\prod_{t\geq 1}P_{V_{i}(t)|V_{i}^{t-1},\underline{V}^{t-1}_{\mathcal{K}}}.
$
%\end{small}

\begin{definition}\label{didef}  \cite{quinn2011equivalence}
A \textit{directed information} graph (DIG) is a directed graph, $\overrightarrow{G}_{DI}$, over a set of random processes $\underline{V}$. Node $i$ represents the random process $V_{i}$; there is an arrow from $j$ to $i$ for $i,j\in \{1,...,n\}$ if and only if:
\begin{equation}\label{crt}
I(V_{j}\rightarrow V_{i}||\ \underline{V}_{-\{i,j\}})>0.
\end{equation}
\end{definition}

Note that in the definition of DIG, it is assumed that there are no latent processes. Thus as demonstrated in the example below, when a subset of processes is not observable (as in our model), the corresponding DIG may not encode the 1-step causal relationships accurately.

\begin{example}
Consider the following joint dynamics: 
\begin{align*}
&X_1(t)=X_1(t-1)/3+W_1(t),\\
&Z(t)=X_1(t-1),\\ 
&X_2(t)=Z(t-1)/3+W_2(t),
\end{align*}
where $\{W_1, W_2\}$ are independent exogenous noises. The corresponding DIG of this system when all processes are observed is $X_1\rightarrow Z\rightarrow X_2$, and when $Z$ is latent, it is $X_1\rightarrow X_2$. But we know that there is no 1-step functional dependency between $X_1$ and $X_2$.
\end{example}

\begin{definition}
 A joint distribution $P_{\underline{V}}$ is called positive if there exists a reference measure $\phi$ such that $P_{\underline{V}}\ll\phi$ and $\frac{dP_{\small{\underline{V}}}}{d\phi} > 0$.
\end{definition}

\begin{remark}
In addition to requiring no latent processes, DIGs recover the structure correctly when underlying distribution is positive. This is to avoid degenerate cases that arise with deterministic relationships. For instance, suppose $X$ and $Y$ are two random processes such that  $Y=\psi(X)$ for some deterministic function $\psi$. Then $P_{X,Y}$ is not positive since the distribution of $Y$ given $X$ is a point mass.
\end{remark}

Note that our model in (\ref{Maineq}) does not satisfy the non-degeneracy assumption. 
This is because in this model the hidden variables are a deterministic functions of the other processes. 
Yet as we will show next, the 1-step causal structure between the observed processes is unique and recoverable as long as the marginal distribution of the observed processes is positive.

%\begin{assumption}\label{assump1}
%We assume that the marginal distribution of observed processes is positive.
%\end{assumption}

%The authors in \cite{quinn2011equivalence} showed that when the joint distribution of a strict causal\footnote{A system is called strictly causal if the processes at present time are conditionally independent given the past of the system.}  system is positive, the corresponding DIG of the system encodes a unique factorization of the joint distribution in which each process $v$ has only been conditioned on its parents, i.e., $\mathcal{PA}(v)$ \cite{quinn2011equivalence}.

\section{Main Result}
Herein, we introduce our approach for learning the 1-step functional dependencies among the observed variables given their realizations. This approach does not require any prior knowledge about the number of latent process nor functions $\{F_i\}$ and $G$. 

\begin{theorem}
Consider the dynamical system in (\ref{Maineq}) and assume that the marginal distribution of the observed variables $P_{\underline{X}}$ is positive. Then $\alpha_{i,j}(t)=0$ if and only if:
\begin{equation}\label{eqTh1}
I(X_i(t);X_j(t-1)|\underline{X}_0^{t-1}\backslash \{X_j(t-1)\})=0.
\end{equation}
\label{Th1}
\end{theorem}
\begin{proof}
First, we prove that if $\alpha_{i,j}(t)=0$ then (\ref{eqTh1}) holds. Suppose that $X_i(t)$ does not 1-step functionally depend on $X_j(t-1)$. According to (\ref{Maineq}), the latent vector $\vec{Z}(t)$ can be determined recursively as a function of $\underline{X}_0^{t-1}$ and $\vec{Z}(0)$. We denote this by $\vec{Z}(t)=\Psi_t(\underline{X}_0^{t-1},\vec{Z}(0))$.
Therefore, the entropy of $X_i(t)$ given $\underline{X}_0^{t-1}$ will be:
%\begin{small}
\begin{align}\label{eqHX}\notag
H(X_i&(t)|\underline{X}_0^{t-1})=H(F_i(\vec{X}(t-1),\vec{Z}(t-1))+\omega_i(t)|\underline{X}^{t-1}_0),\\ \notag
&= H(F_i(\vec{X}(t-1),\Psi_{t-1}(\underline{X}^{t-2}_0,\vec{Z}(0)))+\omega_i(t)|\underline{X}^{t-1}_0),\\ 
&= H(\omega_i(t)|\underline{X}^{t-1}_0)=H(\omega_i(t)).
\end{align}
%\end{small}
The last equation holds because $F_i(\vec{X}(t-1),\Psi_{t-1}(\underline{X}_0^{t-2},\vec{Z}(0)))$ is a deterministic function of  $\underline{X}_0^{t-1}, \vec{Z}(0)$, and $\omega_i(t)$ is independent of $\underline{X}_0^{t-1}$. Furthermore, we have:
\begin{equation}
\begin{split}
H(X_i(t)&|\underline{X}_0^{t-1}\backslash \{X_j(t-1)\})\\
&=H(F_i^{\prime}(X_j(t-1))+\omega_i(t)|\underline{X}_0^{t-1}\backslash \{X_j(t-1)\}),
\end{split}
\end{equation}
where $F_i^{\prime}(X_j(t-1))$ is a uni-variate function obtained from $F_i(\vec{X}(t-1),\Psi_{t-1}(\underline{X}^{t-2}_0,\vec{Z}(0)))$ by determining the values of $\underline{X}_0^{t-1}\backslash X_j(t-1)$. But $F_i^{\prime}(X_j(t-1))$ does not change  by varying $X_j(t-1)$ since we assumed $\alpha_{i,j}(t)=0$. Hence, the above equation is equal to $H(\omega_i(t))$ and by comparing with (\ref{eqHX}), we can deduce that (\ref{eqTh1}) holds.

For the converse, note that $X_i(t)$ and $X_j(t-1)$ are independent given $\underline{X}_0^{t-1}\backslash \{X_j(t-1)\}$ according to (\ref{eqTh1}). Consequently, we have:
\begin{equation}
\mathbb{E}\{X_i(t)|\underline{X}_0^{t-1}\}=\mathbb{E}\{X_i(t)|\underline{X}_0^{t-1}\backslash \{X_j(t-1)\}\}
\label{EX}
\end{equation}

For any realization of $\underline{X}_0^{t-1}$ like $\underline{x}_0^{t-1}$, the left hand side of the above equation is equal to:
\begin{equation}
\mathbb{E}\{X_i(t)|\underline{X}_0^{t-1}=\underline{x}_0^{t-1}\}=F_i(\vec{x}(t-1),\vec{z}(t-1)).
\label{EX2}
\end{equation}
where $\vec{z}(t-1)=\Psi_{t-1}(\underline{x}_0^{t-2},\vec{Z}(0))$. Since the joint distribution of the observed processes $P_{X}$ is positive, we know that $X_j(t-1)$ cannot be written as a deterministic function of $\underline{X}_0^{t-1}\backslash \{X_j(t-1)\}$. Thus, the right hand of (\ref{EX}) does not depends on $X_i(t-1)$. From this fact and (\ref{EX2}), we can conclude that $F_i(\vec{X}(t-1),\vec{Z}(t-1))$ is not a function of $X_j(t-1)$ for any realization of $\underline{X}_0^{t-1}$ and thus $\alpha_{i,j}(t)=0$.
\end{proof}

This result can be used to recover the 1-step causal structure of the observed processes in (\ref{Maineq}) given $X_0^t$. To do so, one can estimate the conditional mutual information in (\ref{eqTh1}) for all $t$. If (\ref{eqTh1}) holds, then we declare that there is no 1-step dependency from $X_j$ to $X_i$.
Next, we propose an efficient method to learn the 1-step causal structure of the observed processes in (\ref{Maineq}) when $F_i$s and $G$ are linear functions.

\subsection{The Linear Model}
Suppose $\{F_i\}$s and $G$ are linear functions, then the equations in (\ref{Maineq}) can be rewritten as follows:
\begin{equation}
\begin{bmatrix}
\vec{X}(t)\\\vec{Z}(t)
\end{bmatrix}
=\begin{bmatrix}
A_{11}& A_{12}\\
A_{21} & A_{22}
\end{bmatrix}
\begin{bmatrix}
\vec{X}(t-1)\\
\vec{Z}(t-1)
\end{bmatrix}
+
\begin{bmatrix}
\vec{\omega}(t)\\
\vec{0}_{m\times 1}
\end{bmatrix},
\label{Varmodel}
\end{equation}
where $[A_{11}]_{n\times n}$,  $[A_{12}]_{n\times m}$, $[A_{21}]_{m\times n}$, and $[A_{22}]_{m\times m}$ denote the coefficient matrices. We also define $A=[A_{11},A_{12};A_{21},A_{22}]$.
%This equation suggests that the 1-step functional dependency between the observed processes is captured by the support of the $A_{11}$. Hence, learning the 1-step causal structure among the observed processes can be done by learning the support of $A_{11}$. In the rest of this section, we propose an efficient method based on linear regression to learn $A_{11}$ given $\underline{X}_0^t$.
The functional dependency of state vector $\vec{X}(t)$ on its history $\underline{X}_0^{t-1}$ and also $\vec{\omega}(t)$, and $\vec{Z}(0)$ for $t>1$ can be written as follows:
\begin{equation}
\begin{split}
\vec{X}(t)&=\displaystyle\sum_{k=0}^{t-1} A^*_k\vec{X}(t-1-k) + A_{12}A_{22}^{t-1} \vec{Z}(0)+\vec{\omega}(t), 
\end{split}
\label{EqVARre}
\end{equation}
where $A_0^*=A_{11}$ and $A_k^*=A_{12}A_{22}^{k-1}A_{21},k\geq 1$. Now, suppose that information-theoretic criteria in (\ref{eqTh1}) is zero. By the same arguments in the proof of Theorem \ref{Th1}, we can show that the following term is zero:
\begin{equation}
\begin{split}
&\mathbb{E}\{X_i(t)|\underline{x}_0^{t-1}\}-  \mathbb{E}\{X_i(t)|\underline{x}_0^{t-1}\backslash \{x_j(t-1)\}\}=\\
&=A_{11}(i,j)\left(x_j(t-1)-\mathbb{E}\{X_j(t-1)|\underline{x}_0^{t-1}\backslash \{x_j(t-1)\}\}\right)
\end{split}
\end{equation}
for any realization of $\underline{X}_0^{t-1}=\underline{x}_0^{t-1}$. Since $P_X$ is positive, we can deduce that $A_{11}(i,j)=0$. Consequently, learning the 1-step causal structure among the observed processes reduces to determining the support of $A_{11}$.
%Hence, we can deduce that $X_i(t)$ is $1$-step functional dependent on $X_j(t-1)$ iff $A_{11}(i,j)\neq 0$. 

Assume that the support of $A_{22}$ corresponds to an acyclic directed graph, i.e. there exists an $l>0$ such that  $A_{22}^l=0$. Under this condition, the equation (\ref{EqVARre}) can be simplified as:
\begin{equation}
\begin{split}
\vec{X}(t)=&\displaystyle\sum_{k=0}^l A^*_k\vec{X}(t-1-k)+\vec{\omega}(t),\quad  t\geq l+1.
\end{split}
\label{ARl}
\end{equation}

The above equation can be interpreted as a VAR model of order $l+1$. Hence, all matrices $\{A_k^*\}$ can be obtained by doing multivariate least square estimation \cite{Lutkepohl}. Moreover, coefficients in the VAR model can be checked for zero constraints by Wald test \cite{Lutkepohl}. Thus, we can check the information-theoretic criteria merely performing a Wald test.

\vspace{.3cm}

\section{Experimental Results}\label{sec:exp}
In this section, we utilize the method described in previous part for network identification problem in consensus protocols \cite{nabi2014network}. In control systems, a well-known approach for network identification is based on running a series of ``node-knockout" experiments in which variables are sequentially forced to be zero without being removed from the network \cite{nabi2014network,shahrampour2015topology}. The main drawback of this approach is that we need to intervene in the system.  Here, we will show that the direct edges between observed nodes can be detected just by analyzing the time-series of observed processes.

Consider the weighted consensus protocol within a system with $n+m$ nodes:
\begin{equation}
V_i(t)=w_{ii}V_i(t-1)+\displaystyle\sum_{j\neq i} w_{ij} (V_j(t-1)-V_i(t-1))+B_i \omega_i(t),
\end{equation}
where $V_i(t)$ represents the state of node $i$ at time $t$ such as its speed, heading, or position, and the weight $w_{ij}$ denotes the weight on the edge $(i,j)$. The first $n$ state variables correspond to states of observed nodes and the rest belong to hidden nodes. We are trying to find all directed edges (with nonzero weight) between observed nodes by injecting the white noise $\omega_i(t)$ into the observed node $i$, i.e. $B_i=1$ if $V_i$ is an observed node and $B_i=0$ otherwise. In fact, this problem can be reformulated to the form in (\ref{Varmodel}) such that:
\begin{equation}
A(i,k)=\begin{cases}
w_{ik},& \mbox{if $i\neq k$},\\
\begin{small} w_{ii}-\sum_{j\neq i} w_{ij}\end{small} ,& \mbox{otherwise,}
\end{cases}
\end{equation}
% $A(i,j)=w_{ij}$ if $i\neq j$ and $A(i,i)=w_{ii}-\displaystyle\sum_{(i,j)\in \vec{E}} w_{ij}}$.
where $[V_1(t),\cdots,V_n(t)]=[X_1(t),\cdots,X_n(t)]$, and $[V_{n+1}(t),\cdots,V_{n+m}(t)]=[Z_1(t),\cdots,Z_m(t)]$.
Hence, identifying all directed edges with nonzero weight between observed nodes is equivalent to obtaining the support of matrix $A_{11}$. 

We generated 1000 instances of the linear system with $n=10$ observed nodes and $m=10$ latent nodes. The weight $w_{ij}$ $(i\neq j)$ was selected randomly from the set $\{-b,0,b\}$ with probability $[q,1-2q,q]$ where $q=0.1$ and $b=0.7$ if $i,j$ were hidden. Otherwise, the weight $w_{ij}$ $(i\neq j)$ was chosen randomly from the set $\{-a,0,a\}$ with probability $[p,1-2p,p]$ where $a=0.2$. Moreover, we set $w_{ii}$ to $\sum_{(i,j)\in \vec{E}} w_{ij}$.

In our simulations, we excluded the generated networks which had cycles in the latent part. Furthermore, the noise process $\vec{\omega}(t)$ was chosen as i.i.d $\mathcal{N}(\vec{0}_{n\times 1},\sigma^2 I_{n\times n})$ with $\sigma^2=0.1$. It can be easily seen that the conditional mutual information in (\ref{Maineq}) is equal to:
\begin{equation}
\begin{split}
I(X_i(t);X_j(t-1)&|\underline{X}_0^{t-1}\backslash X_j(t-1))=\\
&=\frac{1}{2}\log \left( 1+\frac{(A_{11}(i,j))^2\sigma_{\omega_j}^2}{\sigma_{\omega_i}^2}\right), t> l
\end{split}
\label{eqsig}
\end{equation}
where $\sigma_{\omega_i}^2$ and $\sigma_{\omega_j}^2$ are the variances of $\omega_i(t)$ and $\omega_j(t)$, respectively. Thus, learning $1$-step functional dependencies, corresponds to finding the support of matrix $A_{11}$, denoted by $Supp(A_{11})$.

In order to obtain nonzero entries in matrix $A_{11}$, we performed a linear regression between $\vec{X}(t+1)$ and data $\underline{X}_{t-l+1}^t$ where $l$ is the lag length. Let $\hat{A}_{11}$ be the output of linear regression for time series of length $10000$. According to Wald test \cite{Lutkepohl}, for large number of samples, we can obtain $Supp(\hat{A}_{11})$ by setting entry $(i,j)$ to one if $|\hat{A}_{11}(i,j)|>a/2$. In Fig. \ref{figs1}, the error $|| Supp(\hat{A}_{11})-Supp(A_{11})||^2_F$ is averaged over generated random matrices where $||.||_F$ is the Frobenius norm of a matrix. As it can be seen, the support of matrix $A_{11}$ can be recovered perfectly as the lag length increases. This trend is expected since the lag length should be at least equal to the order of linear model $l$, in order to have perfect recovery. Moreover, as shown in Fig. \ref{figs1}, for a fixed lag length, the average error is higher for larger $p$. This is because the matrices $\{A_k^*\}$ become more dense for larger $p$ which leads to higher average error when the right lag length is not selected.

\begin{figure}[!t]
\centering
\includegraphics[width=3.5in]{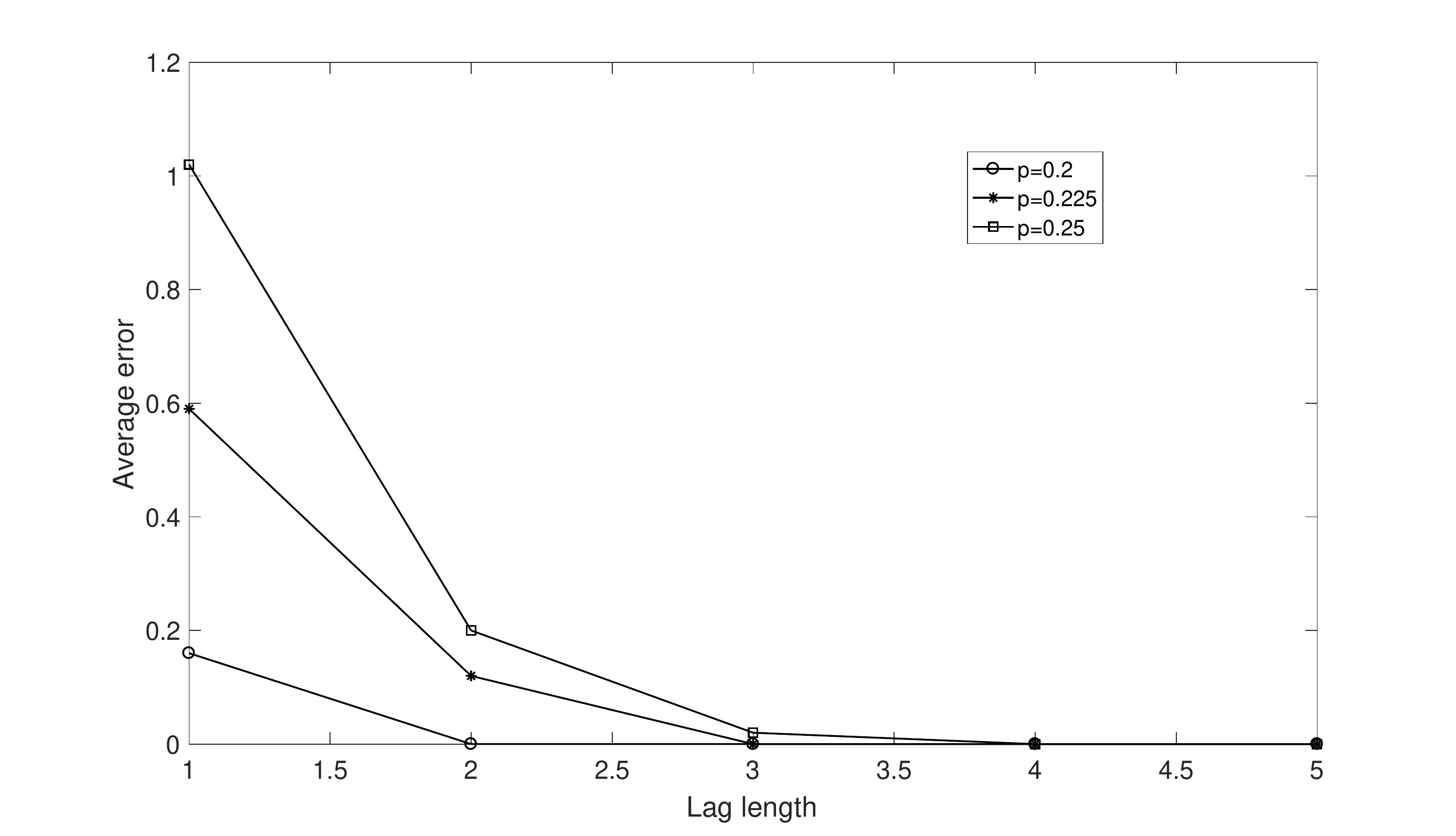}
\caption{Average error versus lag length for different values of parameter $p$.}
\label{figs1}
\end{figure}

We also examined our proposed criteria in a nonlinear system with three state variables with the following joint dynamics:
\begin{eqnarray}\label{Nonlinear}
\left\{
\begin{array}{lll}
X_1(t)&=0.2X_1(t-1)+0.4\sqrt{|Z(t-1)|}+\omega_1(t),\\
X_2(t)&=0.5(X_1(t-1))^2+0.9Z(t-1)+\omega_2(t),\\
Z(t)&=0.9(X_1(t-1))^3+0.4Z(t-1),
   \end{array} 
\right. 
\end{eqnarray}
where $\omega_1(t)$ and $\omega_2(t)$ are i.i.d $\mathcal{N}(0,0.1)$ and $X_1(0),X_2(0)$ have normal distribution with zero mean and unit variance. The quantity in (\ref{Maineq}) can be written as a linear combination of some joint entropies. Hence, we can utilize the $K$-nearest neighbor method of \cite{Singh} to obtain an estimation of the desired quantity. To do so, we generated $1000$ samples of $X_1,X_2,Z$ for $t=0,1$. For $K=10$, the numerical results were: $I(X_1(1);X_2(0)|X_1(0))\approx 0.06$, and $I(X_2(1);X_1(0)|X_2(0))\approx 1.68$. From these results, we can infer that $X_2(t)$ is 1-step functional dependent on $X_1(t-1)$ which is consistent with the system model in (\ref{Nonlinear}).

\vspace{.3cm}

\section{Conclusion}
We proposed an information-theoretic quantity for identifying causal relations among observed variables for a general $1$-step stochastic dynamical system in the presence of latent variables when there exists no exogenous noise in the latent part. It would be interesting to see if by further imposing some additional constraints on the structure of functional dependencies, it would be possible to recover the inter-connections in the latent sub-graph.

\bibliographystyle{plain}
\bibliography{ref}

\end{document}